\documentclass[acmsmall,nonacm,screen]{acmart}
\usepackage{braket}
\usepackage{amsmath}
\usepackage{amsfonts}
\usepackage{xcolor}
\usepackage{appendix}
\usepackage{booktabs}
\usepackage{cleveref}

\usepackage[ruled,linesnumbered]{algorithm2e}
\newcommand{\Supp}{\mathrm{Supp}}

\newcommand{\blue}[1]{#1}

\newcommand{\norm}[1]{\left\Vert#1\right\Vert}

\begin{document}

\title{Optimal quantum sampling on distributed databases}
\author{Longyun Chen}
\affiliation{
  \department{QICI Quantum Information and Computation Initiative, Department of Computer Science, School of Computing and Data Science}
  \institution{The University of Hong Kong}
  \city{Hong Kong SAR}
  \country{China}
}
\affiliation{
  \institution{Nanjing University}
  \department{State Key Laboratory for Novel Software Technology}
  \city{Nanjing}
  \country{China}
}
\email{chenlongyun01@gmail.com}

\author{Jingcheng Liu}
\affiliation{
  \department{State Key Laboratory for Novel Software Technology, New Cornerstone Science Laboratory}
  \institution{Nanjing University}
  \city{Nanjing}
  \country{China}
}
\email{liu@nju.edu.cn}

\author{Penghui Yao}
\affiliation{
  \department{State Key Laboratory for Novel Software Technology, New Cornerstone Science Laboratory}
  \institution{Nanjing University}
  \city{Nanjing}
  \country{China}
}

\affiliation{
  \institution{Hefei National Laboratory}
  \city{Hefei}
  \country{China}
}
\email{phyao1985@gmail.com}



\begin{abstract}
%Quantum sampling, which encodes a given probability distribution in the amplitudes of a pure state, plays a crucial role as a subroutine in numerous quantum algorithms.  This paper focuses on quantum sampling when the data is distributed among multiple databases and each database solely maintains a basic oracle that counts the multiplicity of individual elements. A coordinator makes oracles to all databases to realize quantum sampling. We specifically investigate the oblivious communication model, where the communication between the coordinator and the database is predetermined. We present both sequential and parallel algorithms: the sequential algorithm queries the databases sequentially, while the parallel algorithm allows the processor to query all the databases simultaneously. We further prove that both algorithms are optimal in their respective settings.

Quantum sampling, a fundamental subroutine in numerous quantum algorithms, involves encoding a given probability distribution in the amplitudes of a pure state.  Given the hefty cost of large-scale quantum storage, we initiate the study of quantum sampling in a distributed setting. Specifically, we assume that the data is distributed among multiple machines, and each machine solely maintains a basic oracle that counts the multiplicity of individual elements.  Given a quantum sampling task, which is to sample from the joint database, a coordinator can make oracle queries to all machines. We focus on the oblivious communication model, where communications between the coordinator and the machines are predetermined. We present both sequential and parallel algorithms: the sequential algorithm queries the machines sequentially, while the parallel algorithm allows the coordinator to query all machines simultaneously. Furthermore, we prove that both algorithms are optimal in their respective settings.

%terminology: should we say "data", or a multiset of elements, or a multiset of computational basis?
\end{abstract}

\maketitle

\section{Introduction}
Quantum sampling is a fundamental computational task in quantum computing that encodes a given distribution in the amplitudes of a quantum state. More specifically, the algorithm has access to a distribution $(p_1,\ldots,p_N)$ and is supposed to output the state $\sum_{i=1}^N\sqrt{p_i}\ket{i}$, where $\{\ket{1},\ldots,\ket{N}\}$ is a set of computational bases. Quantum sampling was inspired by the famous Grover's algorithm~\cite{10.1145/237814.237866} and is nowadays a key subroutine in many quantum algorithms. For example, the well-known Harrow-Hassidim-Lloyd algorithm~\cite{PhysRevLett.103.150502}, which solves a system of linear equations $Ax=b$ with an exponential speedup over the fastest classical algorithm, requires encoding the vector $b$ to the amplitude of a pure state $\ket{b}=\sum_i b_i\ket{i}$ up to normalization.

Many quantum algorithms for learning functions and distributions also require quantum sampling on a given distribution~\cite{gilyen_et_al:LIPIcs.ITCS.2020.25,Arunachalam2021twonewresultsabout,10.5555/3291125.3309633}. It is also known that the quantum advantages of certain quantum learning algorithms require quantum sampling and the advantages would vanish if quantum sampling was replaced by classical sampling~\cite{gilyen_et_al:LIPIcs.ITCS.2020.25}. Moreover, quantum sampling has also found many algorithmic applications, such as quantum walk~\cite{doi:10.1137/090745854,1366222,PhysRevA.78.042336}, quantum mean estimation~\cite{10.1145/3519935.3520045,hamoudi:hal-03454632,hamoudi:hal-02349991}, and quantum coupon collector~\cite{arunachalam_et_al:LIPIcs.TQC.2020.10}.  Thus, a number of works have been devoted to designing algorithms and analyzing the complexities of quantum sampling~\cite{5959806,10.1145/2493252.2493256,6108195}.

However, due to current physical limitations, having a single quantum storage for big data remains challenging. In order to achieve large-scale and more fault-tolerant quantum sampling, we initiate the study of  {\em distributed quantum sampling}. We assume that datasets are distributed across multiple machines, both for reducing the storage complexity for a single machine, and enabling fault-tolerance in the databases.
The main overhead introduced by distributed datasets are therefore the communication required between different machines.
As such, we primarily focus on the round complexity of (quantum) communications to achieve distributed quantum sampling.
It is important to note that our motivation to introduce "distributedness" into quantum sampling is to overcome the current physical limitations of large-scale quantum storage, rather than to accelerate computations as in traditional distributed computation. To carry out quantum sampling in a distributed setting adds complexity for algorithm and system designers. 

Our model of distributed database consists of several machines and a coordinator. The distributed database has a publicly known maximum capacity $\nu$ for each kind of element, which is an upper bound on multiplicities of the elements. %With the oracle design declared later, the parameter $\nu$ is needed to bound the dimension of the register. 
Intuitively, this bound $\nu$ is needed for encoding the multiplicities using quantum registers, so that different machines can communicate via quantum communications.
A large dataset is distributed across these machines. The goal of the coordinator is to produce the quantum state $\sum_{i=1}^N\sqrt{p_i}\ket{i}$, where $p_i$ is the probability that you get $i$ when sampling uniformly from the distributed database. 

To minimize the implementation cost of each database, we assume that each database only needs to implement a simple oracle operation that maintains the number of times an element appears in its share of the dataset, which is a common setting for the quantum database in the previous works~\cite{10.1145/237814.237866,Boyer_1998,JuYi-Lin2007QCDa,Liu_2023}. Indeed, the assumption that the database only permits a specific type of query is generally accepted in theory of quantum computing, particularly within the context of quantum query complexity. This assumption is motivated by practical considerations -- the realization of a quantum database that accepts a fixed type of queries is arguably more feasible than the implementation of a fully fledged quantum database. This perspective stems from technical obstacles in the current development and maintenance of an all-powerful quantum database. In light of this, we only allow quantum communications in the form of a simple oracle query.  %to solve the quantum sampling problem, the coordinator can only communicate with the machines through the oracle queries. 

We elaborate the technical challenges in designing distributed databases that support quantum sampling.
\paragraph{Necessity of quantum communications}
To illustrate, we discuss a simplified setting where there is a single machine holding an unknown subset S of elements in $[N]$, and the coordinator needs to guess another subset of elements T that overlaps with S substantially. This is a special case of the distributed sampling problem because, failing to do so would not produce a sample with constant fidelity. If we only allow classical communications, then it requires $\Omega(N)$ in communication complexity % (either classical or quantum)
because of the existence of error correcting codes with linear rate. This means that the coordinator has to effectively ask every database, how many times every possible element appear.
For $n$ machines and a dataset defined on a data universe (alphabet) of size $N$, this means the classical query complexity could be as large as $nN$.

\paragraph{Physical and practical restrictions in quantum communications}
To reduce communication complexity, we allow quantum communications between the coordinator and machines. For the purpose of illustration, we revisit the simplified setting where there is a single machine, and consider a quantum variant of it: suppose that the machine is holding an unknown state $\sum_{i\in S} \ket{i}$, and the coordinator is trying to come up with another state $\sum_{i\in T} \ket{i}$ that approximates it well.  With unrestricted access to quantum communications, this single-machine quantum sampling problem becomes trivial because the machine could simply send the state over to the coordinator in one round. However, it requires an all-powerful quantum machine to prepare the state for the coordinator. 

Another issue arises when one tries to generalize the above idea to quantum sampling with more than one machine.
A natural attempt would be to run a quantum sampling algorithm on each machine individually, and let the coordinator combine their output somehow. Specifically, each machine simply produce its own quantum sampling state, after which the coordinator attempts to combine them to create a global sampling state. Unfortunately, this is not a physically realizable scheme, because combining quantum sampling states is not a quantum unitary operation\footnote{For example, suppose there are two machines, each holding a different element, $x$ and $y$. The quantum sampling state of each machine is $\ket{x}$ or $\ket{y}$. The coordinator aims to output the state $ ( \ket{x} + \ket{y})/\sqrt{2}$. An operator that takes input $\ket{x}\ket{y}$ and outputs $( \ket{x} + \ket{y})/\sqrt{2}  $ for every pair of states $\ket{x}$ and $\ket{y}$ cannot be a linear operator, even with ancillaries.}. Therefore, such an attempt is not possible using valid quantum operations, even with all-powerful quantum communications and zero approximation error from each database. 
%As explained in our previous rebuttal, the coordinator cannot combine these states together to produce a global quantum sampling state, even with zero approximation error from each database. 

To address both the practical and physical considerations, we take inspiration from arguably the most successful generic quantum search framework to date, the Grover search framework~\cite{10.1145/237814.237866}, and let each machine implement a local Grover oracle. The specific form of the query is carefully designed so that the coordinator can combine the individual quantum queries from multiple machines, using only unitary transformations.

%If we only allow classical communications, the coordinator has to send queries to each database, asking the multiplicity of every possible element. After the coordinator has learned $p_i$ for every $i$, then it has to prepare the quantum state by itself. For a dataset with a data universe of $N$ elements containing all distinct elements distributed across $n$ machines, the query complexity could be as large as $nN$. We study the quantum query complexity where we allow quantum communications between the coordinator and the machines, and we are able to show significant speed-up.

\subsection*{Main results}
In this paper, we exhibit two distributed quantum sampling algorithms using {\em sequential} and {\em parallel} queries, respectively. Informally speaking, in the sequential model, the coordinator makes queries to each database sequentially. In the parallel model, the coordinator queries all machines simultaneously. We further show that both algorithms achieve optimal query complexity among all oblivious algorithms.

Our algorithms allow different machines to hold the same key because this is a more general setting, which does not require synchronization of quantum distributed databases (a costly assumption in classical settings). Our lower bound holds even if all databases are disjoint. Intuitively, the coordinator does not know which elements are in which machine and has to establish this knowledge with confidence through sufficiently many queries.

\begin{theorem}[Main result, informal]
    Given a distributed database consisting of $n$ machines with maximum capacity $\nu$, there exists an algorithm for quantum sampling with $O(n\sqrt{\nu N/M})$ sequential queries in the oblivious model. If parallel queries are allowed, then $O(\sqrt{\nu N/M})$ queries are sufficient. Here $N$ is the size of the data universe, and $M$ is the total number of elements stored across the distributed database counting multiplicities. Moreover, both algorithms are optimal in the oblivious communication model.
\end{theorem}

These two optimal algorithms are designed by directly expanding the centralized quantum sampling algorithm. Their optimality is established by extending the method proposed by Zalka for Grover's algorithm in the centralized setting~\cite{Zalka_1999}. Our proof suggests that the essential barrier for distributed quantum sampling is the same as that for the centralized setting. %This fact reveals that there is no essential gap between the centralized and distributed databases for quantum sampling.

\subsection*{Related work}
The problem of quantum sampling was raised after Grover's algorithm~\cite{10.1145/237814.237866}. The quantum query complexity of quantum sampling has been studied in various contexts. Shi~\cite{1181975} introduced the problem of index erasure: given an injective function $f:[n]\rightarrow[m]$ via a black-box oracle, the task is to prepare the quantum state, which is the uniform superposition on the image of $f$, i.e., $\sum_{x=1}^n\ket{f(x)}/\sqrt{n}$. Index erasure can be viewed as a uniform quantum sampling over a subset of the universe. This problem is closely related to graph isomorphism, and the tight query complexity of the problem was later established by Ambainis, Magnin, Roetteler, and Roland in the coherent setting~\cite{5959806} and by Lindzey and Rosmanis in the non-coherent setting~\cite{lindzey_et_al:LIPIcs.ITCS.2020.59}. Ozols, Roetteler, and Roland~\cite{10.1145/2493252.2493256} further introduced quantum rejection sampling, which converts a quantum state to another quantum state with a given amplitude. The quantum query complexity of quantum state conversion has been established in~\cite{6108195}. 

In addition to quantum query complexity, quantum sampling has also been studied in other models of computation. Aharonov and Ta-Shama~\cite{aharonov2003adiabaticquantumstategeneration}~studied the problem of preparing $\sum_{i\in\Omega}\sqrt{p_i}\ket{i}$ given the description of a classical circuit with output distribution $p$. A weaker quantum sampling model, where an extra register is allowed, that is, $\sum_i\sqrt{p_i}\ket{i}\ket{c_i}$, has also been considered in~\cite{hamoudi:hal-02349991,10.5555/3291125.3309633}. 

\paragraph*{Organization} \Cref{sec:preli} contains preliminary materials on basic notations. \Cref{sec:models} formally introduces the model of quantum distributed databases and quantum sampling on distributed databases. \Cref{sec:alg} presents both sequential and parallel algorithms for quantum sampling. \Cref{sec:lowerbound} proves lower bounds on query complexity in both sequential models and parallel models, which shows that both algorithms in \Cref{sec:alg} are optimal.

\section{Preliminaries}\label{sec:preli}
For integer $N>0$, let $[N]$ represent the set $\{1,\cdots,N\}$. Given a multiset $S$, $\Supp(S)$ represents the support of $S$. For any element $x$, the multiplicity of $x$ is the number of occurrences in $S$. The cardinality of a multiset $S$, denoted by $|S|$, is the sum of the multiplicities of all its elements.

Here we give a brief introduction to quantum computing and the notations used in this paper. Readers may refer to \cite{nielsen2010quantum} for a thorough treatment. Consider a Hilbert space $\mathcal{H}$ endowed with an inner product $\langle\cdot,\cdot\rangle$. A quantum state $\rho$ is a positive semidefinite matrix with a trace equal to $1$. 

Let $\ket{\psi}$ be a vector in $\mathcal{H}$. The norm of $\ket{\psi}$, denoted by $\norm{\ket{\psi}}$ is defined to be $\norm{\psi}:=\sqrt{\bra{\psi}\psi\rangle}$. For any two vectors $\ket{\phi}$ and $\ket{\psi}$ in $\mathcal{H}$, the distance between them is $\norm{\ket{\phi}-\ket{\psi}}.$ A quantum register $A$ is associated with a Hilbert space $\mathcal{H}_A$. The composition of two registers $A$ and $B$, denoted by $AB$, is associated with the Hilbert space $\mathcal{H}_A\otimes\mathcal{H}_B$. The identity operator on $\mathcal{H}_A $, (and associated register $A$) is denoted $I_A$. The subscript $A$ may be omitted when it is clear from the context. 

\blue{When the rank of $\rho$ is $1$, the state $\rho = \ket{\psi}\bra{\psi}$ is pure and can be represented by $\ket{\psi}$. For a mixed state $\rho = \sum_i p_i\ket{\psi_i}\bra{\psi_i}$ with $\sum_i p_i=1$, it can be regarded as an ensemble where the state $\ket{\psi_i}$ is observed with probability $p_i$. Given two quantum states $\rho,\sigma$, we measure their similarity by quantum fidelity, defined as
\[
F(\rho,\sigma) = \left(\operatorname{Tr}\left[\sqrt{\sqrt{\rho}\sigma\sqrt{\rho}}\right]\right)^2 \in [0,1].
\]
When $\rho = \ket{\psi}\bra{\psi},\sigma = \ket{\phi}\bra{\phi}$ are pure, the fidelity becomes $|\braket{\psi|\phi}|^2$, the square of the inner product.}

\section{Distributed databases}\label{sec:models}
In this section, we formally introduce the model of distributed databases considered in this paper. 
A distributed database consists of several machines, each of which stores part of the data and is maintained by a machine. Moreover, each database also implements some simple operations. There is a coordinator who makes queries to each database and outputs the answer at the end of the algorithm. The coordinator is assumed to be a quantum computer. In this paper, we assume that the coordinator sends an element from the dataset to a database, and the database answers the multiplicity of the element, i.e., the number of occurrences of the element in the database. A mathematical formulation is given below.  We are interested in minimizing the number of quantum queries made by the coordinator.

In this paper, we only consider the {\em oblivious} communication model, where the order of the communication between the coordinator and the machines is predetermined (only depends on the public knowledge known to the coordinator). The oblivious communication model has been studied in~\cite{Le_Gall_2022}.

We conjecture that non-oblivious communication does not help us to save the number of queries. It is worth noting that the final output of the algorithm is supposed to be a pure state. Thus, in the quantum circuits model, all intermediate measurements can be removed by the principle of deferred measurement and the gentle measurement argument~\cite{nielsen2010quantum,10.1145/3313276.3316378}. However, it is not clear whether they can be extended to a distributed setting. We leave it for future work.

\subsection*{Quantum sampling on distributed databases}

Suppose that the data universe is represented by the set $[N]:=\{1,\cdots, N\}$, and the dataset is distributed among $n$ machines.
The coordinator maintains a quantum state with three registers
\[\ket{\rho}=\sum_{i\in[N]}\alpha_i\ket{i}\ket{s_i}\ket{w_i}.\]
The first register is $N$-dimensional for element storage, the second register is $(\nu+1)$-dimensional to store the outcome of the oracle, and the last one is the ancillary register, whose dimension remains to be determined by the algorithm design. In our algorithm, $w_i$ should belong to $\{0,1\}$.

In this paper, we consider two models of queries: {\em sequential} queries and {\em parallel} queries.
In the sequential model, the coordinator sends queries to the machines sequentially.  In a sequential model, 
suppose the coordinator makes a query to the $j$-th database. It sends the first two registers to the $j$-th database. The $j$-th database implements the following operation $\mathcal{O}_j$:
\begin{equation}\label{eqn:sequentialO}
\mathcal{O}_j\ket{i}\ket{s}=\ket{i}\ket{(s+c_{ij}) \mod (\nu+1)},
\end{equation}
where $c_{ij}$ is the multiplicity of an element $i$ in the $j$-th database and $\nu$ is the maximum capacity of each database, which is known to the coordinator. Thus, $\nu\geq \max_{i\in[N]}(\sum_{j=1}^nc_{ij})$ is an upper bound for the multiplicities of the elements. 

It is worth noting that the oracle operation can be easily extended to a dynamic database. It is low-cost to update oracle operation $\mathcal{O}_j$ if the datasets are changed. For instance, if the multiplicity of element $i$ in the $j$-th database increases or decreases by $1$, i.e., $c_{ij}$ increases or decreases by $1$, we can simply update $\mathcal{O}_j$ by left multiplying operator $U$ or $U^{\dagger}$, respectively, where $U\ket{i}\ket{s}=\ket{i}\ket{(s+1)~\mod{(\nu+1)}}$.

In the parallel model, the coordinator may send multiple queries to distinct machines simultaneously. To be more specific, in the parallel model, the state with the coordinator contains four registers 
\[\ket{\rho}=\sum_{\bar{i}\in[N]^n}\alpha_{\bar{i}}\ket{\bar{i}}\ket{s^{\bar{i}}}\ket{b^{\bar{i}}}\ket{w_{\bar{i}}},\]
where $\bar{i}=(i_1,\ldots,i_n)\in[N]^n, b^{\bar{i}}\in\{0,1\}^n$.
Thus, each of the first three registers contains $n$ qudits. When the coordinator makes a query, it sends three qudits, one from each of the first three registers, to each database. For $j\in[n]$, the $j$-th database is implementing
\begin{equation}\label{eqn:parallelOj}  \hat{\mathcal{O}}_j\ket{i_j}\ket{s_j}\ket{b_j}=\ket{i_j}\ket{s_j+c_{i_{j},j}\cdot b_j\mod{(\nu+1)}}\ket{b_j}
\end{equation}
where $\nu$ is an upper bound on the multiplicities of $i_j$.
It is not hard to see that the operation in \Cref{eqn:parallelOj} can be realized by the query operation in the sequential model defined in \Cref{eqn:sequentialO}. Thus a parallel query 
\begin{equation}\label{eqn:parallelquery}
\mathcal{O}\ket{\bar{i}}\ket{s_1\cdots s_n}\ket{b_1\cdots b_n}=\bigotimes_{j=1}^n\hat{\mathcal{O}}_j\ket{i_j,s_j,b_j}
\end{equation}
can be implemented by $n$ sequential queries. 

To describe the problem of quantum sampling, we need to further introduce some notations. 
For $j\in[n]$, the dataset on the $j$-th machine is denoted by a multiset $T_j$.  Then, the multiset $T_j$ is completely determined by the values of $c_{ij}$ defined in \Cref{eqn:sequentialO}. 
% Specifically, the support set is given by $\Supp(T_j)=\{i\in[N]|c_{ij}>0\}$ and the multiplicity function is simply $i\mapsto c_{ij}$.

%\liuexp{TODO: explain why multiplicity makes sense.}

Now we are ready to formally define the problem of quantum sampling on distributed databases. Given datasets $\{T_j\}_{j\in [n]}$, a quantum sampling algorithm is supposed to produce the state
\begin{equation}\label{eq:psi}
    \ket{\psi}=\frac{1}{\sqrt{M}}\sum_{i\in[N]}\sqrt{c_i}\ket{i},
\end{equation} 
where $c_i=\sum_{j\in[n]}c_{ij}$ is the total occurrences of the element $i$ across all machines and $M=\sum_{i\in[N]}c_i$ is the total count of all elements on the machines. 
Since $c_i/M$ is the frequency of the element $i$ appearing over all machines, measuring the state $\ket{\psi}$ under the computational basis $\{\ket{i}\}_{i\in[N]}$ is equivalent to sampling over the datasets.

All parameters and notations are summarized in Table~\ref{tb:parameters}.

\begin{table*}
    \caption{Table of Notations}\label{tb:parameters}
    \begin{tabular}{cc}
        \toprule
        Symbol & Meaning \\
        \midrule
        $n$ & the count of the machines \\
        $N$ & the number of the varieties of elements \\
        $T_j$ & the dataset (multiset) on the $j$-th machine \\
 $c_{ij}$ &the multiplicity of element $i$ in $T_j$ \\
 $c_{i}:=\sum_{j\in[n]}c_{ij}$ &the total count of occurrences of $i$ across all machines \\
        $M:=\sum_{i\in[N]}c_i$& the total count of the elements over all machines \\
 $M_j:=|T_j|$&the count of the elements on the $j$-th machine\\
        $\Supp(T_j)$ & the support set of distinct elements in $T_j$\\
        $m_j:=|\Supp(T_j)|$ & the count of distinct elements on the $j$-th machine\\
        $\nu$ & maximum capacity of the database \\
        $\ket{\psi}$ & the quantum sampling state defined in \Cref{eq:psi} \\
        $\ket{\pi}:=\frac{1}{\sqrt{N}}\sum_{i\in[N]}\ket{i}$ & the uniform superposition state\\
        $\mathcal{O}_j$ & the oracle of the $j$-th machine\\
        \bottomrule
    \end{tabular}
\end{table*}

% When the parallel query is allowed, by Algorithm \ref{alg:implementation of parallel query}, it suffices to show that the composition of two oracles $O_jO_k$, for $j\neq k$, can be processed simultaneously. In other words, to apply $O_{j_1}\cdots O_{j_\tau}$ for distinct $j_1,\cdots,j_\tau\in[n]$, the time cost is $O(\tau)$ within the ordinary model, but it is reduced to only a constant within the parallel model.

% Formally speaking, with an extra control register $\ket{b}=\ket{b_1,\cdots,b_n}$ to determine whether $O_j$ should be applied, a parallel query can be described by an oracle
% \begin{equation}\label{eq:parallel oracle}
% \mathcal{O}\ket{i}\ket{c}\ket{b}:=\left[\left(\prod_{b_j=1}O_j\right)\ket{i}\ket{c}\right]\otimes\ket{b}.
% \end{equation}
% For the parallel model, to reflect the time saved by the parallel operation, the query complexity should be defined as the count of the parallel query, i.e., the number of the oracle $\mathcal{O}$ invoked, instead of the number of ordinary oracles $O_j$.

\section{Algorithms}\label{sec:alg}
At the heart of our algorithm, each machine implements a local Grover step by the oracle, which allows us to query each database one by one while effectively simulating a global amplitude amplification. This minimizes the quantum ability of each machine, and its simplicity allows it to be easily updated when the database is dynamically changing.
\subsection{Sequential queries}
We start by giving a quantum sampling algorithm for the sequential model. A key operator in our algorithm is the {\em distributing operator} $D$ such that
\begin{equation}\label{definition of D}
    D\ket{i,0}=\sqrt{\frac{c_i}{\nu}}\ket{i,0}+\sqrt{\frac{\nu-c_i}{\nu}}\ket{i,1},
\end{equation}
We claim that there exists a unitary operator satisfying the above equation, and is thus a valid quantum operation.

\begin{lemma}
    The operator $D$ can be extended to a unitary operator on the whole Hilbert space. 
    %Furthermore, the operator $Q$ also can be extended to a unitary.
\end{lemma}
\begin{proof}
    \Cref{definition of D} defines the operator $D$ on the domain of a subspace spanned by $\{\ket{i,0}\}_{i\in[N]}$. It can be checked that \[\braket{i,0|D^\dagger D|j,0}=\delta_{ij}\]
with the Kronecker notation $\delta_{ij}$. Thus, $D$ preserves the inner product on the subspace, which can be extended to a unitary operator on the whole space.
%sNotice that $Q$ is the product of some unitary operator, which is unitary as well. 
\end{proof}

Notice that the operator $D$ depends on the input $\{T_j\}$ due to its definition in \Cref{definition of D}. The next lemma shows that the operator $D$ can be realized by $2n$ calls of $\mathcal{O}_j$'s.
\begin{lemma}\label{prop:implementation of D}
    The operator $D$ can be implemented with $2n$ queries given by the oracles $\mathcal{O}_j$ defined in \Cref{eqn:sequentialO} and unitary operators independent of the input.
\end{lemma}
\begin{proof}
The implementation of the operator $D$ is given by the following three steps:
\begin{align*} \ket{i,0,0}&\xrightarrow{\mathcal{O}_1\cdots \mathcal{O}_n\otimes I}\ket{i,c_i,0}
    \xrightarrow{\ \mathcal{U}\ }\sqrt{\frac{c_i}{\nu}}\ket{i,c_i,0}+\sqrt{\frac{\nu-c_i}{\nu}}\ket{i,c_i,1} \\
    &\xrightarrow{\mathcal{O}_1^\dagger\cdots \mathcal{O}_n^\dagger\otimes I}\sqrt{\frac{c_i}{\nu}}\ket{i,0,0}+\sqrt{\frac{\nu-c_i}{\nu}}\ket{i,0,1}.
\end{align*}
The first and third steps can be realized by queries to $n$ machines. The unitary operator $\mathcal{U}$ is defined to satisfy
\begin{equation}\label{eq:definition of mathcal{U}}
	\mathcal{U}\ket{i,c,0}=\sqrt{\frac{c}{\nu}}\ket{i,c,0}+\sqrt{\frac{\nu-c}{\nu}}\ket{i,c,1},
\end{equation}
which is independent of the input. It is not hard to see
\[\braket{i,c,0|\mathcal{U}^\dagger\mathcal{U}|i',c',0}=\delta_{(i,c),(i',c')}.\] Thus $\mathcal{U}$ is a unitary operator.

\end{proof}
\begin{theorem}[Quantum sampling with sequential queries] \label{thm:query complexity within the sequential model}
    Given parameters $\varepsilon\in(0,1)$, and  $N,M,n,\nu$ as in \Cref{tb:parameters} satisfying $\nu\geq \frac{M}{N\varepsilon}$, there exists an algorithm for quantum sampling which makes $O\left(n\sqrt{\nu N/M}\right)$ queries and outputs the quantum sampling state $\ket{\psi}$ defined in \Cref{eq:psi}.
\end{theorem}
\begin{proof}
Recall the operator $D$ in~\Cref{definition of D}.
By direct calculation, we have 
\begin{equation}\label{definition of phi}
    \begin{aligned}
        D\ket{\pi,0}&=\frac{1}{\sqrt{N}}\sum_{i\in[N]}\left(\sqrt{\frac{c_i}{\nu}}\ket{i,0}+\sqrt{\frac{\nu-c_i}{\nu}}\ket{i,1}\right)\\
        &=\sqrt{\frac{M}{\nu N}}\ket{\psi,0}+\sqrt{1-\frac{M}{\nu N}}\ket{\psi^{\bot},1},
    \end{aligned}
\end{equation}
where $\ket{\psi}$ is the target state defined in~\cref{eq:psi}, and $\ket{\psi^{\bot}, 1}$ is a pure state orthogonal to $\ket{\psi, 0}$. Since the amplitude $\sqrt{\frac{M}{\nu N}}$ of the target state is known, we can applying zero-error amplitude amplification~\cite[Theorem 4]{Brassard_2002} with $O(\sqrt{\nu N/M})$ calls of $D$, then the final state will be $\ket{\psi,0}$.

Specifically, choose a unitary operation $S_\chi(\varphi)$ defined by
\[
    \ket{i}\ket{b} \longmapsto \begin{cases}
        \mathrm{e}^{\mathrm{i}\varphi} \ket{i}\ket{b} & b=0 \\
        \ket{i}\ket{b} & \text{otherwise}
    \end{cases}
\]
and $S_\pi(\phi)$ defined by
\[
    F\ket{i} \otimes \ket{b} \longmapsto \begin{cases}
        \mathrm{e}^{\mathrm{i}\phi} F\ket{i} \otimes \ket{b} & i,b=0 \\
        F\ket{i} \otimes \ket{b} & \text{otherwise}
    \end{cases}
\]
where $F$ is the quantum Fourier transforms used to prepare the initial state $\ket{\pi}=F\ket{0}$. Let $\theta = \arcsin \sqrt{\frac{M}{\nu N}}$, $\tilde{m} = \frac{\pi}{4\theta} - \frac{1}{2}$, and
\[
    Q(\varphi,\phi) = -D S_\pi(\phi) D^\dagger S_\chi(\varphi).
\]
After applying $Q(\pi,\pi)$ with $\varphi=\phi=\pi$ a number of $\lfloor \tilde{m} \rfloor$ times to $D\ket{\pi,0}$ and then applying $Q(\varphi,\phi)$ once, the final state is $\ket{\psi,0}$ if and only if~\cite[Equation 12]{Brassard_2002}
\[
    \cot \left((2\lfloor\tilde{m}\rfloor+1) \theta\right)=\mathrm{e}^{\mathrm{i}\varphi} \sin \left(2 \theta\right)\left(-\cos \left(2 \theta\right)+\mathrm{i} \cot (\phi / 2)\right)^{-1}.
\]
The right side can be any complex number of norm at most $\tan(2\theta)$. Since 
\[
    \frac{\pi}{2} - 2\theta \leq (2\lfloor\tilde{m}\rfloor+1)\theta \leq \frac{\pi}{2},
\]
the left side is not greater than $\tan(2\theta)$. Therefore, with an appropriate choice of $\phi$ and $\varphi$, we can obtain $\ket{\psi,0}$ as the final state.

\end{proof}

\subsection{Parallel queries}
Modifying the algorithm for the sequential model, we then give a sampling algorithm for the parallel model. 
We still adopt the sampling algorithm in \Cref{thm:query complexity within the sequential model} given by amplitude amplification. The only change is the implementation of the operator
\[D:\ket{i,0}\longmapsto \sqrt{\frac{c_i}{\nu}}\ket{i,0}+\sqrt{\frac{\nu-c_i}{\nu}}\ket{i,1}\]
to reduce the query complexity.

\begin{lemma}
    The operator $D$ can be implemented with four queries given by the parallel query oracle $\mathcal{O}$ defined in \Cref{eqn:parallelquery} and unitary operators independent of the input.
\end{lemma}
\begin{proof}
    By the proof of \Cref{prop:implementation of D}, the operator $D$ can be implemented in three steps. The second step is a unitary operator $\mathcal{U}$ independent of the input. Here, we are going to realize the first and the third steps with the oracle $\mathcal{O}$. The first step is $\ket{i,0}\mapsto\ket{i,c_i}$. With a ancillary registers initialized as $\ket{0^n,0^n,0^n}$, this step can be completed by
    \begin{align*}
        \ket{i,0,0^n,0^n,0^n}&\xrightarrow{\qquad\ \ }\ket{i,0,i^n,0^n,1^n}\\
        &\xrightarrow{I\otimes I\otimes\mathcal{O}}\ket{i,0,i^n,c_{i1}c_{i2}\cdots c_{in},1^n}\\
        &\xrightarrow{\qquad\ \ }\ket{i,c_i,i^n,c_{i1}c_{i2}\cdots c_{in},1^n}\\
    &\xrightarrow{I\otimes I\otimes\mathcal{O}^\dagger}\ket{i,c_i,i^n,0^n,1^n}\rightarrow\ket{i,c_i,0^n,0^n,0^n}.
    \end{align*}
    This procedure only applies $\mathcal{O}$ twice. The third step is just the inverse of the first step, which can be completed similarly. 
\end{proof}

\begin{theorem}[quantum sampling with parallel queries]
    Given parameters $\varepsilon\in(0,1)$, and  $N,M,n,\nu$ as in \Cref{tb:parameters} satisfying $\nu\geq \frac{M}{N\varepsilon}$, there exists an algorithm for quantum sampling which makes $O\left(\sqrt{\nu N/M}\right)$ parallel queries and outputs the quantum sampling state $\ket{\psi}$.
\end{theorem}
\begin{proof}
	Similarly to the proof of \Cref{thm:query complexity within the sequential model}, the algorithm begins with the uniform superposition state $\ket{\pi}$ and applies zero-error amplitude amplification \cite[Theorem 4]{Brassard_2002} with $O(\sqrt{\nu N/M})$ calls of $D$, obtaining a final state $\ket{\psi,0}$.
\end{proof}

% The query complexity is reduced by the new implementation of $D$. 
% \begin{theorem}[Query complexity within the parallel model]
%     Given parameters $N,M,n,\nu$, the sampling algorithm (Algorithm \ref{sampling algorithm for the parallel model}) possesses a query complexity of $O\left(\sqrt{\frac{\nu N}{M}}\right)$.
% \end{theorem}
% \begin{proof}
% 	Suppose $D$ needs $t_D$ calls to the oracles. Similarly to the proof of Theorem \ref{thm:query complexity within the oblivious model}, the complexity equals to $\left(2\lfloor\frac{\pi}{4\theta}\rfloor+1\right)t_D$. With the new implementation of the operator $D$ (Algorithm \ref{alg:implementation of D within the parallel model}), it holds that $t_D=1$. This implies the complexity is $\left(2\lfloor\frac{\pi}{4\theta}\rfloor+1\right)=O(1/\theta)$. Since $\sin\theta=\sqrt{\frac{M}{\nu N}}\geq \frac{2}{\pi}\theta$ for $\theta\in[0,\pi/2]$, the query complexity is $O\left(\sqrt{\frac{\nu N}{M}}\right)$.
% \end{proof}

\section{Optimality}\label{sec:lowerbound}

In this section, we prove that our algorithms are optimal. Specifically, we prove a lower bound for the query complexity for the quantum sampling problem, which aligns with the complexity of our algorithms. Intuitively, we design a potential function $D_t$ that captures the “confidence of knowledge” that the coordinator has. We show that without sufficiently many queries, the potential function will be large, leading to a small correlation between an adversarially chosen ground truth and the output state. Its formal description is given as \Cref{lem:lower bound for D_t} and \Cref{lem:upper bound for D_t} for the sequential queries, and \Cref{lem:lower bound for D_t with in the parallel model} and \Cref{lem:upper bound for D_t within the parallel model} for the parallel queries.
 
Here, we consider a slightly more general setting where each database may have an independent maximum capacity $\kappa_j$. Thus for each $j\in[n]$, it holds  
\[\max_{i\in[N]}c_{ij}\leq\kappa_j\leq \nu.\] For the case where $\kappa_j$ is unknown, we can just use $\nu$ for $\kappa_j$.

And the $j$-th machine maintains quantum oracles $\mathcal{O}_j$ and $\hat{\mathcal{O}}_j$ with the functionality
%\[O_j\ket{i}\ket{c}=\ket{i}\ket{(c+c_{ij}) \mod (\kappa_j+1)}.\]
%We further assume that the $j$-th machine is implementing a more powerful oracle operator
\begin{align*}
\mathcal{O}_j\ket{i}\ket{s}&=\ket{i}\ket{(s+c_{ij}) \mod (\nu+1)},\\
    \hat{\mathcal{O}}_j \ket{i,s,b}&=\begin{cases}
        (\mathcal{O}_j\ket{i,s})\otimes\ket{b}, & b=1, \\
        \ket{i,s,b}, & b=0.
    \end{cases}
\end{align*}

In addition, although our proposed algorithm in \Cref{sec:alg} can output the target state $\ket{\psi}$ with a zero error, the proof here considers a larger class of algorithms, which only need to obtain an approximate state $\rho$ with a fidelity greater than a constant.

% The output of an algorithm is the state of the first register. Denote the Hilbert space for the first register by $\mathcal{X}$, which is spanned by $\{\ket{i}\}_{i\in[N]}$. And the space for the other registers is $\mathcal{Y}$. 

\begin{theorem}[Lower bound for the sequential model]\label{Lower bound for the query complexity}
   For any sampling algorithm in an oblivious communication model with sequential queries, if it satisfies that $F(\rho,\psi)>9/16$, where $\rho$ is the output state and $\psi$ is the quantum sampling state, then its query complexity is at least $\Omega\left(\sum_{j\in[n]}\sqrt{\frac{\kappa_j N}{M}}\right)$.
\end{theorem}

\begin{theorem}[Lower bound for the parallel model]\label{thm:parallel}
    For any sampling algorithm in an oblivious communication model with parallel queries, if it satisfies that $F(\rho,\psi)>9/16$, where $\rho$ is the output state and $\psi$ is the quantum sampling state, then its query complexity is at least $\Omega\left(\max_{j\in[n]}\sqrt{\frac{\kappa_j N}{M}}\right)$.
\end{theorem}

\subsection{Oblivious queries with measurement}

We start by observing that, given an oblivious algorithm with measurements for quantum sampling, one can construct an algorithm without any measurements that has the same query complexity. Therefore, we can focus on algorithms without measurements in the rest of the paper.

\begin{lemma}\label{lem:measurement_to_no_measurement}
	Let $\mathcal{A}$ be an oblivious algorithm with measurements for quantum sampling. Then there exists an algorithm $\mathcal{B}$ without measurements, which has the same query complexity and fidelity.
\end{lemma}
It can be proved by the deferred measurement principle with minor modification. The proof is deferred in \Cref{sec:proofof1}.

\subsection{Hard inputs}

We prove the optimality via the quantum adversary argument~\cite{AMBAINIS2002750}. To this end, we describe the hard inputs in this subsection.

Let $t_k$ be the number of times the oracle $\hat{\mathcal{O}}_k$ and $\hat{\mathcal{O}}_k^\dagger$ is applied; then the query complexity  is $\sum_{k\in[n]}t_k$. We prove a lower bound for each $t_k$. Notice that for an oblivious model, the order of queries is independent of the inputs. Thus, this lower bound applies to all inputs, which implies that the summation of the individual lower bounds is a lower bound for the total query complexity.

Let $\sigma$ be a permutation on $[N]$ and $S\subseteq[N]$ be a subset. We say $\sigma$ is {\em order-preserving} for $S$ if for any $r,t\in S$, it holds that $\sigma(r)<\sigma(t)$ if and only if $r<t$. 

We fix any integer $k$ for the rest of this section.
Given a sequence of multisets $T=(T_1, T_2, \cdots, T_n)$ and an order-preserving $\sigma$ for $\Supp(T_k)$,
we permute $T_k$ by $\sigma^{-1}$ to obtain a new sequence of multisets $T'=(T_1, T_2, \cdots,T_k',\cdots, T_n)$. Specifically, define \[
c'_{ij}=\begin{cases}
	c_{ij}, & j\neq k,\\
	c_{\sigma^{-1}(i)j}, & j=k,
\end{cases}
\]
where $c_{ij}$ are the multiplicities for $T_j$. Notice that $c'_{ij}$ uniquely define a sequence of $\{T_j'\}$.
We write $\tilde{\sigma}^k(T): = \{T_j'\}$, and call $\tilde{\sigma}_k$ as the $\sigma$-induced permutation.

\begin{definition}[Hard input condition]\label{def:hardinputcond}
    Given $k\in [n]$, constants $\alpha,\beta\in (0,1]$, and a sequence of multisets $T=(T_1, T_2, \cdots, T_n)$ distributed on $n$ machines, then $T$ satisfies the hard input condition  if 
    \begin{equation}\label{eq:constraint on hard input}
	M_k\geq \alpha M, \quad M_k/m_k\geq \beta \kappa_k, \quad \max_{i\in[N],j\neq k}c_{ij}+\max_{i\in[N]}c_{ik}\leq \nu,
\end{equation}
where $M_k=|T_k|$, $m_k=|\Supp(T_k)|$ and $c_{ij}$ is the multiplicity of $i$ in $T_j$. 
\end{definition}
\begin{definition}[Hard inputs]\label{def:hardinput}
Given $k\in [n]$, constants $\alpha,\beta\in (0,1]$, and a sequence of multisets $T=(T_1, T_2, \cdots, T_n)$ satisfying the hard input condition in~\Cref{def:hardinputcond},
the collection of hard inputs for the $k$-th machine is defined as 
\[\mathcal{T}:=\{\tilde{\sigma}^k(T) : \sigma~\text{is order-preserving for }\Supp(T_k)\},\]
where $\tilde{\sigma}^k$ is the $\sigma$-induced permutation defined above.
\end{definition}

The last condition in \Cref{def:hardinputcond} guarantees $\tilde{\sigma}_k(T)\in\mathcal{T}$ is still a hard input with multiplicities not greater than $\nu$.

The following lemma gives the size of the collection of hard inputs.
\begin{lemma}\label{lem:size_of_T}
	For any $k\in[N]$, constants $\alpha,\beta\in (0,1]$, and a sequence of multisets $T=(T_1, T_2, \cdots, T_n)$ satisfying the hard input condition, let $\mathcal{T}$ be the collection of hard inputs as given in \Cref{def:hardinput}. It holds that $|\mathcal{T}|=\binom{N}{m_k}$ with $m_k:=|\operatorname{Supp}(T_k)|$.
\end{lemma}
\begin{proof}
	Let $S=\Supp(T_k)$. A claim should be concluded to calculate the size of $\mathcal{T}$.
	
	For $\sigma_1,\sigma_2$ that are order-preserving for $S$, we claim that $\tilde{\sigma}^k_1(T)=\tilde{\sigma}^k_2(T)$ if and only if $\sigma_1(i)=\sigma_2(i)$ for every $i\in S$. The sufficiency is obvious. For the necessity, we prove by contradiction. Suppose $\tilde{\sigma}^k_1(T)=\tilde{\sigma}^k_2(T)$ and $\sigma_1(i_0)\neq\sigma_2(i_0)$ for some $i_0\in S$. Consider the multiplicity $c'_{\sigma_1(i_0)k}$ for $T':=\tilde{\sigma}^k_1(T)=\tilde{\sigma}^k_2(T)$, it follows that $c_{i_0k}=c_{\sigma_2^{-1}(\sigma_1(i_0))k}>0$. Denote $\sigma_2^{-1}(\sigma_1(i_0))$ by $i_1$. Since $\sigma_1(i_0)\neq\sigma_2(i_0)$, it holds that $i_0\neq i_1$. If $i_0<i_1$, then the order-preserving property of $\sigma_1,\sigma_2$ implies
	\begin{align*}
		|\{i>\sigma_1(i_0)|c'_{ik}>0\}| &= |\{i>i_0|c_{ik}>0\}| \\
		&> |\{i>i_1|c_{ik}>0\}| \\
		&=|\{i>\sigma_2(i_1)=\sigma_1(i_0)|c'_{ik}>0\}|,
	\end{align*} 
	The inequality holds because $i_1$ belongs to the former set but not the latter set. The first expression and the last expression are the same, which leads to a contradiction. It is similar for the case of $i_0>i_1$.
	
	With this claim, the size $|\mathcal{T}|$ equals the count of order-preserving permutations that act differently on $S$. Finding such a permutation is equivalent to choosing $|S|$ elements in $[N]$ as the image set $\sigma(S)$. Hence, we have
	$|\mathcal{T}|=\binom{N}{|S|}=\binom{N}{m_k}$,
	as $|S|=|\operatorname{Supp}(T_k)|=:m_k$.
\end{proof}

\subsection{Lower bound on sequential queries}
We are now ready to derive a lower bound on the query complexity. As argued above, we can bound each $t_k$ independently. To do so,
 for every $k\in [n]$, we consider a collection of hard inputs $\mathcal{T}$ generated by an input $T$ with respect to $k$, as given in \Cref{def:hardinput}.

Given an input $T$, let $\ket{\psi^T_t}$ be the state after $t$ calls to the oracle. It can be expressed as 
\begin{equation}\label{eqn:psiTt}
    \ket{\psi_{t}^T}=U_{t}O_t U_{t-1}O_{t-1}\cdots U_1 O_{1}U_0\ket{0},
\end{equation}
where $O_1,\cdots,O_t$ are either $\hat{\mathcal{O}}_k\otimes I$ or $\hat{\mathcal{O}}_k^\dagger\otimes I$ with identity operator $I$ on the registers that $\hat{\mathcal{O}}_k$ does not act on, and $U_0,\cdots,U_t$ are unitary operators that are independent of $T_k$, the datasets on the $k$-th machine. 
We consider an input $\tilde{T}$ obtained from $T$ by removing the dataset on the $k$-th machine. That is, we replace $T_k$ with an empty set, and the datasets on other machines are the same as those in $T$. Notice that, each of the $\{O_i\}$ is an identity operator if the dataset on the $k$-th machine is empty. Thus, for the dataset $\tilde{T}$, the state after $t$ calls of the oracle is 
\begin{equation}\label{eq:psi_t}
	\ket{\psi_{t}}=U_{t} U_{t-1}\cdots U_1U_0\ket{0}.
\end{equation}
%We note that every $T\in \mathcal{T}$ share the same $\tilde{T}$ and therefore the same $\ket{\psi_{t}}$, so 
We introduce a potential function as follows:
\begin{equation}\label{eqn:DT}
   D_t=\frac{1}{|\mathcal{T}|}\sum_{T\in\mathcal{T}}\left\| \ket{\psi_t^T}-\ket{\psi_t}\right\|^2. 
\end{equation}

We note that for our collection $\mathcal{T}$, while the state $\ket{\psi_{t}^T}$ depends on the specific choice of $T\in\mathcal{T}$, the state $\ket{\psi_{t}}$ remains the same regardless of $T$. To see this, note that for any pair $T,T' \in \mathcal{T}$, $T$ and $T'$ only differs in $T_k$.  Further, due to the obliviousness of queries, the oracles not involving $\hat{\mathcal{O}}_k$ or $\hat{\mathcal{O}}_k^\dagger$ remain the same for $\tilde{T}$ (and hence for $T$ and $T'$), regardless of $T_k$. Therefore, any input $T$ and $T'$ that only differs in $T_k$ share the same state $\ket{\psi_{t}}$.

To elaborate our proof ideas for the lower bound, we first note that $\ket{\psi_{t_k}^T}$ is the final state of the algorithm, and it must hold that the state of the first register in $\ket{\psi_{t_k}^T}$ approximates the goal state $\ket{\psi}$. Informally speaking, since the goal state $\ket{\psi}$ also depends on $T_k$, which is different for different $T \in \mathcal{T}$, any sampling algorithm must bring a large enough difference between $\ket{\psi^T_{t_k}}$ and $\ket{\psi_{t_k}}$ to get $\ket{\psi}$ for every $T\in\mathcal{T}$. Through this intuition, uniformly picking a $T\in\mathcal{T}$, we consider the expectation of the variation 
\begin{equation}\label{eqn:Dt}
    D_t=\mathbb{E}_\mathcal{T}\left[\left\|\ket{\psi_t^T}-\ket{\psi_t}\right\|^2\right].
\end{equation}
This idea is formally depicted in \Cref{lem:lower bound for D_t}.
\begin{lemma}\label{lem:lower bound for D_t}
    Let $\mathcal{T}$ be the collection of hard inputs for the $k$-th machine as in \Cref{def:hardinput}. Let $\alpha,\beta$ be the constants defined in \Cref{def:hardinput}. Suppose the fidelity between the output state $\rho$ and the quantum sampling state $\psi$ defined in \Cref{eq:psi} satisfies $F(\rho,\psi)\geq(1-\epsilon)^2>9/16$ with $\epsilon\geq 0$. If $M<\beta^2\kappa_k N/16$ and $\alpha>4\epsilon$, then the expectation of the variation defined in \Cref{eqn:Dt} is bounded by 
    $D_t\geq C\frac{M_k}{M}$ 
   for some constant $C$ dependent to $\alpha$ and $\epsilon$.
\end{lemma}

\blue{The increment of $D_t$ only comes from the application of $\hat{\mathcal{O}}_k$ and $\hat{\mathcal{O}}_k^\dagger$. We can estimate the maximum variation of $D_t$ with a single oracle call, leading to an upper bound of $D_t$ with respect to the time $t$ of oracle calls, as stated in \Cref{lem:upper bound for D_t}. To achieve a value of $D_t$ that is at least as large as the lower bound given in \Cref{lem:lower bound for D_t}, it is necessary for $t$ to be sufficiently large to ensure that the upper bound is not smaller than the lower bound.}
\begin{lemma}\label{lem:upper bound for D_t}
    Let $\mathcal{T}$ be the collection of hard inputs for the $k$-th machine as in \Cref{def:hardinput}. For $t\leq t_k$, it holds that $D_{t_k}\leq 4\frac{m_k}{N}t^2$.
\end{lemma}

\blue{The proof of \Cref{lem:lower bound for D_t} and \Cref{lem:upper bound for D_t} is deferred to \Cref{sec:proof of lower bound for D_t} and \Cref{sec:proof of upper bound for D_t}, respectively. With these two lemmas, we are ready to show a lower bound for the query complexity.}

\begin{proof}[Proof of \cref{Lower bound for the query complexity}]
    We start by showing that for each $j \in [n], t_j=\Omega\left(\sqrt{\kappa_j N/M}\right)$. Fixing a $k\in[n]$, for a constant $\beta\in(0,1]$, we split the proof into two cases of $M\geq \beta^2\kappa_k/16 N$ and $M<\beta^2\kappa_k N/16$ .

    For $M\geq \beta^2\kappa_k N/16$, since the model is oblivious, nothing about $T_k$, except $\kappa_k$, is known for the coordinator before the end of the algorithm. If $\kappa_k=0$, clearly we have $t_k\ge 0$ . If $\kappa_k>0$, then $T_k$ is non-empty, and we have to invoke the oracle corresponding to the $k$-th machine to get information about $T_k$. Thus, the oracle $\hat{\mathcal{O}}_k$ should be applied at least once in principle when $\kappa_k>0$. In this way, it holds that $t_k\geq 1 \geq (\beta/4)\cdot\sqrt{\kappa_k N/M}$.

    For $M<\beta^2\kappa_k N/16$, we consider an input $T$ for $k$ with constants $\alpha,\beta$ as in \Cref{def:hardinput}. Let $\alpha\in(4\epsilon,1]$. Since $M<\beta^2\kappa_k N/16$, we can put all of the elements to the $k$-th machine to construct the input $T$ satisfying the hard input conditions in \cref{eq:constraint on hard input}. Thus, the hard input $T$ must exist. By Lemma \ref{lem:lower bound for D_t}, it follows that $D_{t_k}\geq CM_k/M$ for $\mathcal{T}$ generated by $T$. Combining Lemma \ref{lem:upper bound for D_t} with $t=t_k$, we have $4m_kt_k^2/N\geq CM_k/M$. Recall that $M_k/m_k\geq \beta\kappa_k$ for hard input by  \Cref{def:hardinput}, it holds that 
    \begin{equation}\label{tk}
        t_k\geq \sqrt{\frac{C}{4}\frac{\beta\kappa_k N}{M}}.
    \end{equation}
    
    Combining these two cases, we have $t_k\geq C'\sqrt{\kappa_k N/M}$ for some positive constant $C'$. Since $k$ is arbitrary, this lower bound holds for $t_j$ for every $j\in[n]$. By obliviousness of the queries, the value of $t_j$ is invariant across every input with the same parameters $N,M,\kappa_j,n$, so we can directly add them together to conclude that the query complexity is
    \[\sum_{j\in[n]}t_j\geq C'\sum_{j\in[n]}\sqrt{\frac{\kappa_j N}{M}}=\Omega\left(\sum_{j\in[n]}\sqrt{\frac{\kappa_j N}{M}}\right).\]
\end{proof}

\subsection{Lower bound on parallel queries}
This subsection proves a lower bound on the query complexity for the parallel model. Similarly to the method for the sequential model, for each $k\in[n]$, we consider the number of oracle calls required for hard inputs for $k$, respectively. Suppose the lower bound obtained by considering the hard inputs for $k$ is $\hat{t}_k$, then $\max_{j\in[n]}\hat{t}_j$ is a lower bound for the query complexity. Since the algorithm has no measurement, for any input $T$, the state after $t$ oracles can be written as
\[\ket{\psi_{t}^T}=U_{t} O_t U_{t-1}O_{t-1}\cdots U_1 O_{1}U_0\ket{0},\]
where $O_1,\cdots,O_t$ are either $\mathcal{O}\otimes I$ or $\mathcal{O}^\dagger\otimes I$ with identity operator $I$ on the registers that $\mathcal{O}$ doesn't act on, and $U_0,\cdots,U_t$ are unitary operators that are independent of the input. We also consider the input $\tilde{T}$ obtained from $T$ by removing the dataset on the $k$-th machine, and the state with input $\tilde{T}$. Suppose the parallel oracle corresponding to $\tilde{T}$ is $\tilde{\mathcal{O}}$, then the state after $t$ calls of this oracle is
\[\ket{\psi_{t}}=U_{t}\tilde{O}_t U_{t-1}\tilde{O}_{t-1}\cdots U_1\tilde{O}_1U_0\ket{0}\]
where $\tilde{O}_1,\cdots,\tilde{O}_t$ are either $\tilde{\mathcal{O}}\otimes I$ or $\tilde{\mathcal{O}}^\dagger\otimes I$

With the above assumptions, we can conclude two lemmas for the parallel query similar to Lemma \ref{lem:lower bound for D_t} and Lemma \ref{lem:upper bound for D_t} for the sequential query.

\begin{lemma}\label{lem:lower bound for D_t with in the parallel model}
    Let $\mathcal{T}$ be the hard input for the $k$-th machine as in \Cref{def:hardinput}. Let $\alpha,\beta$ be the constants defined in \Cref{def:hardinput}. Suppose the fidelity between the output state $\rho$ and $\psi$ satisfies $F(\rho,\psi)\geq(1-\epsilon)^2>9/16$ with $\epsilon\geq 0$. If $M<\frac{\beta^2}{16}\kappa_k N$ and $\alpha>4\epsilon$, then  
    \[\mathbb{E}_\mathcal{T}\left[\left\|\ket{\psi_{\hat{t}_k}^T}-\ket{\psi_{\hat{t}_k}}\right\|^2\right]\geq C\frac{M_k}{M}\] 
   for some constant $C$ dependent to $\alpha$ and $\epsilon$.
\end{lemma}
\begin{proof}
	Noticing the proof of Lemma \ref{lem:lower bound for D_t} is independent of the form of the oracle, it also holds for the parallel queries. The lemma is then obtained immediately.
\end{proof}

\begin{lemma}\label{lem:upper bound for D_t within the parallel model}
    For every collection $\mathcal{T}$ generated by a hard input for $k$ and $t\leq \hat{t}_k$, it holds that $\mathbb{E}_\mathcal{T}\left[\left\|\ket{\psi_t^T}-\ket{\psi_t}\right\|^2\right]\leq 4\frac{m_k}{N}t^2$ within the parallel model.
\end{lemma}
\begin{proof}
	The proof of Lemma \ref{lem:upper bound for D_t} depends on the form of the oracle. But, specifically, the only part of it that depends on the oracle is the proof of Proposition \ref{sum O-I}. Hence, we only need to prove the conclusion of this proposition 
	\[\sum_{T\in\mathcal{T}}\left\|\left(O_{t+1}-\tilde{O}_{t+1}\right)\ket{\psi_t}\right\|^2\leq 4\frac{m_k}{N}\binom{N}{m_k}\]
	within the parallel model.
	
	Expanding the identity operator,
	\[
    \begin{aligned}
        &(O_{t+1}-\tilde{O}_{t+1})\ket{\psi_t}= \\
        &\sum_{\bar{i}:i_k\in T_k}\sum_{\bar{s}\in \{0,\dots,\nu\}^n}\sum_{b\in\{0,1\}^n}\sum_l\ket{\bar{i},\bar{s},b,l}\bra{\bar{i},\bar{s},b,l}\left(O^\dagger_{t+1}-\tilde{O}^\dagger_{t+1}\right)\ket{\psi_t}.
    \end{aligned}
    \] 
	we have
	\begin{align*}
	&\sum_{T\in\mathcal{T}}\left\|\left(O_{t+1}-\tilde{O}_{t+1}\right)\ket{\psi_t}\right\|^2 \\
    \leq& 4\sum_{T\in\mathcal{T}}\sum_{\bar{i}:i_k\in T_k}\sum_{\bar{s}\in\{0,1,\cdots,\nu\}^n}\sum_{b\in\{0,1\}^n}\sum_l |\braket{\bar{i},\bar{s},b,l|\psi_t}|^2.
	\end{align*}
	Similarly to the proof of Proposition \ref{sum O-I}, changing the order of the summation, we can obtain the same expression as Proposition \ref{sum O-I}. This gives the lemma.
\end{proof}

Following the same idea for the sequential query, we can conclude the lower bound for the parallel query complexity with these two lemmas.

\begin{proof}[Proof of \Cref{thm:parallel}]
	Similarly to the proof of \Cref{Lower bound for the query complexity}, by the above two lemmas, we can conclude that $\hat{t}_k\geq C'\sqrt{\frac{\kappa_k N}{M}}$
	for some constant $C'$. Hence, the query complexity is not less than
	$\max_{j\in[n]}\hat{t}_j=\Omega\left(\max_{j\in[n]}\sqrt{\frac{\kappa_j N}{M}}\right)$.
\end{proof} 
\section{Conclusion}
\blue{Motivated by the hefty cost of a single large-scale quantum storage, we design distributed databases for a fundamental task in quantum computing - quantum sampling. In our design, data can be partitioned across multiple machines with limited local quantum storage. By introducing basic Grover oracles maintained by each machine, we discuss the distributed quantum sampling task in two query patterns: sequential query and parallel query.}

\blue{With the proposed quantum distributed model, in an oblivious query setting, we show that a sequential sampling algorithm requiring $\Theta(n\sqrt{\nu N/M})$ queries and a parallel sampling algorithm requiring $\Theta(\sqrt{\nu N/M})$ queries by extending the method for the centralized model. These results demonstrate that the fundamental limitations of distributed quantum sampling align with those of the centralized case, despite the additional constraints imposed by the distributed architecture.}

\blue{This work opens several promising directions for future research. The extension to non-oblivious communication models is still possible to yield further improvements in query complexity. Additionally, only a simple distributed architecture is discussed here, it remains valuable to consider other types of architecture close to the practical scenario for a quantum network.}

\paragraph{Acknowledgment.} LC was supported by the National Natural Science Foundation of China (NSFC)/Research Grants Council (RGC) Joint Research Scheme (Grant No. N\_HKU7107/24), the National Natural Science Foundation of China via the Excellent Young Scientists Fund (Hong Kong and Macau) (Grant No. 12322516) and the Hong Kong Research Grant Council (RGC) through the General Research Fund (GRF) grant (Grant No. 17303923). JL and PY were supported by National
Natural Science Foundation of China (Grant No. 62332009, 12347104, 62472212), Innovation Program
for Quantum Science and Technology (Grant No. 2021ZD0302901), NSFC/RGC Joint Research
Scheme (Grant No. 12461160276), Natural Science Foundation of Jiangsu Province (Grant No.
BK20243060) and the New Cornerstone Science Foundation.

\bibliographystyle{ACM-Reference-Format}
\bibliography{ref}
\appendix
\section{Proof of Lemma~\ref{lem:measurement_to_no_measurement}}\label{sec:proofof1}

	By the assumption of the oblivious algorithm, the order of the oracles that $\mathcal{A}$ makes is predetermined and independent of the input, which, therefore, is also independent of the outcomes of the measurements. Thus, the measurement can be deferred to the end of the algorithm. Without loss of generality, it can be assumed that there is only one projective measurement \cite{nielsen2010quantum}.
	
	Suppose this projective measurement is described by the projection operators $\{\Pi_i\}_{i=1}^{N}$. The algorithm now implements $\mathcal{V}$ and then follows a projective measurement. Suppose the state before this measurement is $\ket{s_{\sigma}}$. Then the output state is 
	\[\rho=\operatorname{Tr_\mathcal{Y}}\left[\sum_{i=1}^{N}\Pi_i\ket{s_\sigma}\bra{s_\sigma}\Pi_i\right],\]
	where $\operatorname{\mathcal{Y}}$ is the non-output registers. The fidelity for $\mathcal{A}$ is
	\[F(\rho,\psi)=\sum_{i=1}^{N}\sum_{\eta} |(\bra{\psi}\otimes\bra{\eta})\Pi_i\ket{s_\sigma}|^2,\]
	where $\{\ket{\eta}\}$ is a basis of $\mathcal{Y}$.
	
	To give a new algorithm $\mathcal{B}$, we use an ancillary bit and choose a unitary transform $U$ such that
	\[U:\ket{s,0}\longmapsto \sum_{i=1}^{N}\sqrt{p_i}\ket{s_i,i}\]
	with an arbitrary state $\ket{s}$, the coefficients $p_i=\braket{s|\Pi_i|s}$ and the normalized projections $\ket{s_i}=\Pi_i\ket{s}/\sqrt{p_i}$. It can be checked that this definition preserves the inner product, due to $\Pi_i\Pi_j=\delta_{ij}$. Hence, $U$ is well-defined. Let the new algorithm be $\mathcal{B}=U(\mathcal{V}\otimes I)$ with initial state $\ket{\zeta}\ket{0}$, where $\ket{\zeta}$ is the initial state of the algorithm $\mathcal{A}$. Since the queries are not changed, the query complexity of $\mathcal{B}$ is the same as $\mathcal{A}$'s.
	
	Finally, we show the fidelity for $\mathcal{B}$ is still the same. The final state of $\mathcal{B}$ is $\sum_{i=1}^{N}(\Pi_i\ket{s_\sigma})\otimes\ket{i}$. Thus, if the output state is $\rho'$, then the fidelity
	\begin{align*}
		F(\rho',\psi) &= \sum_{\eta}\sum_{l=0}^{N}|\bra{\psi,\eta,l}\sum_{i=1}^{N}(\Pi_i\ket{s_\sigma})\otimes\ket{i}|^2 \\
		&= \sum_{\eta}\sum_{l=1}^{N}|\bra{\psi,\eta}\Pi_l\ket{s_\sigma}|^2 \\
		&= F(\rho,\psi).
	\end{align*}

\section{Proof of Lemma \ref{lem:lower bound for D_t}} \label{sec:proof of lower bound for D_t}

To simplify the expression of the fidelity in the condition, we introduce another pure state $\ket{\tilde{\psi}^T}$.
\begin{lemma} \label{lem:purification for fidelity}
	Let $\mathcal{X}$ and $\mathcal{Y}$ be the output register and the working register of the coordinator, respectively. Let $\rho=\operatorname{Tr_{\mathcal{Y}}}[\ket{\psi^T_{t_k}}\bra{\psi^T_{t_k}}]$ be the output state. If $\dim\mathcal{X}\leq \dim\mathcal{Y}$, then $F(\rho,\psi)=|\braket{\psi^T_{t_k}|\tilde{\psi}^T}|^2$ for 
	\begin{equation}\label{eq:psi'}
		\ket{\tilde{\psi}^T}=\frac{1}{\sqrt{M}}\sum_{i\in[N]}\sqrt{c_{i}}\ket{i}\ket{\xi^T_i}\in\mathcal{X}\otimes\mathcal{Y}
	\end{equation}
	with some $\ket{\xi_i^T}\in\mathcal{Y}$. 
\end{lemma}
\begin{proof}
	Since $\ket{\psi^T_{t_k}}$ is a purification of $\rho$, by Uhlmann's theorem, the fidelity \[F(\rho,\psi)=\max_{\ket{v}\in\mathcal{X}\otimes\mathcal{Y}}|\braket{\psi^T_{t_k}\mid v}|^2\]
    with $\ket{v}$ satisfying $\operatorname{Tr_{\mathcal{Y}}}[\ket{v}\bra{v}]=\ket{\psi}\bra{\psi}$.
	Let $\ket{\tilde{\psi}^T}$ be the state $\ket{v}$ that makes the inner product attain the maximum. Suppose $\ket{\tilde{\psi}^T}=\sum_{i\in[N]}\omega_i \ket{i}\ket{\xi_i^T}$. Since $\operatorname{Tr_{\mathcal{Y}}}[\ket{\tilde{\psi}^T}\bra{\tilde{\psi}^T}]=\ket{\psi}\bra{\psi}$, it must hold that $|\omega_i|=\sqrt{c_i/M}$. Moving the phase of $\omega_i$ to the phase of $\ket{\xi_i^T}$, we can suppose that $\omega_i=\sqrt{c_i/M}$, which leads to the lemma.
\end{proof}

Without loss of generality, we suppose that $\dim\mathcal{Y}$ is sufficiently large. Then by Lemma \ref{lem:purification for fidelity}, the condition $F(\rho,\psi)\geq (1-\epsilon)^2$ of Lemma \ref{lem:lower bound for D_t} implies
$|\braket{\psi^T_{t_k}\mid\tilde{\psi}^T}|\geq 1-\epsilon$.

To obtain the lower bound for $D_{t_k}$ given in Lemma \ref{lem:lower bound for D_t}, we divide $D_{t_k}$ into two parts:
\[E_{t_k}=\frac{1}{|\mathcal{T}|}\sum_{T\in\mathcal{T}}\left\|\ket{\psi_{t_k}^T}-\ket{\tilde{\psi}^T}\right\|^2, F_{t_k}=\frac{1}{|\mathcal{T}|}\sum_{T\in\mathcal{T}}\left\|\ket{\psi_{t_k}}-\ket{\tilde{\psi}^T}\right\|^2\]
by a triangle inequality as follows:
\begin{align}
    D_{t_k}&=\frac{1}{|\mathcal{T}|}\sum_{T\in\mathcal{T}}\left\|\ket{\psi_{t_k}^T}-\ket{\tilde{\psi}^T}+\ket{\tilde{\psi}^T}-\ket{\psi_{t_k}}\right\|^2\notag\\
    &\geq E_{t_k}+F_{t_k}-\frac{2}{|\mathcal{T}|}\sum_{T\in\mathcal{T}}\left\|\ket{\psi_{t_k}^T}-\ket{\tilde{\psi}^T}\right\| \left\|\ket{\tilde{\psi}^T}-\ket{\psi_{t_k}}\right\|\notag\\
    &\geq E_{t_k}+F_{t_k}-2\sqrt{E_{t_k}}\sqrt{F_{t_k}}\notag\\
    &=\left(\sqrt{F_{t_k}}-\sqrt{E_{t_k}}\right)^2. \label{eq:lower bound for D_t by E_{t_k} and F_{t_k}}
\end{align}
Hence, we should look for a lower bound for $\sqrt{F_{t_k}}-\sqrt{E_{t_k}}$. In the following part, we upper bound $E_{t_k}$ in Lemma \ref{lem:upper bound for E_{t_k}} and lower bound $F_{t_k}$ in Lemma \ref{lem:lower bound for F_{t_k}}.

\begin{lemma}\label{E}\label{lem:upper bound for E_{t_k}}
    Suppose $M\leq\beta^2\kappa_k N/16$ and $\mathcal{T}$ is a collection of hard inputs for $k\in[n]$ as in~\Cref{def:hardinput}. If the final state $\ket{\psi_{t_k}^T}$ satisfies that $|\braket{\psi^T_{t_k}\mid\tilde{\psi}^T}|\geq 1-\epsilon$ for every $T\in\mathcal{T}$, with $\ket{\tilde{\psi}^T}$ given by  \cref{eq:psi'}, then it holds that $E_{t_k}\leq 2\epsilon$.
\end{lemma}
\begin{proof}
    Since the change of the global phase does not affect the quantum state, without loss of generality, we can assume 
    \[
    \braket{\psi_{t_k}^T\mid \tilde{\psi}^T}=\left|\braket{\psi_{t_k}^T\mid\tilde{\psi}^T}\right|\geq 1-\epsilon.
    \]
    Therefore,
    \[\left\|\ket{\psi_{t_k}^T}-\ket{\tilde{\psi}^T}\right\|^2=2-2\braket{\psi_{t_k}^T\mid\tilde{\psi}^T}\leq 2\epsilon.\]
    It follows that $E_{t_k}\leq 2\epsilon$.
\end{proof}

To bound $F_{t_k}$, we need the following proposition:
\begin{proposition}\label{psi sum}
    Let $\mathcal{T}$ be a collection of hard inputs for $k\in[n]$ with constants $\alpha,\beta$, then \[\sum\limits_{T\in\mathcal{T}}\left|\braket{\psi_{t_k}\mid\tilde{\psi}^T}\right|\leq \sqrt{\frac{\sum_{j\neq k}M_j}{M}}|\mathcal{T}|+\sqrt{\frac{\kappa_k}{MN}}m_k|\mathcal{T}|,\]
    with $\ket{\psi_{t_k}}$ defined by \Cref{eq:psi_t} and $\ket{\tilde{\psi}^T}$ given by \Cref{eq:psi'}.
\end{proposition}
\begin{proof}
	Let $\mathcal{Z}$ be the space of the output register and working register. 
    Consider an embedding map $\mathcal{Z}\hookrightarrow\mathcal{Z}\otimes\mathbb{C}^n$ defined as $\ket{\varphi}\mapsto\ket{\varphi}\ket{0}$, where $\mathbb{C}^n$ is the $n$-dimensional complex Hilbert space. Define a linear transform $A$ on $\mathcal{Z}\otimes\mathbb{C}^n$ satisfying 
    \begin{equation}
        \begin{aligned}
            A\ket{\tilde{\psi}^T,0}
            =&A\left(\frac{1}{\sqrt{M}}\sum_{i\in[N]}\sqrt{c_i}\ket{i,\xi_i^T,0}\right)\\
            =&\frac{1}{\sqrt{M}}\sum_{i\in[N]}\sum_{j\in[n]}\sqrt{c_{ij}}\ket{i,\xi_i^T,j}, 
        \end{aligned}
        \label{eq:A*psi'}
    \end{equation}
   and
   \begin{equation}
       \begin{aligned}
           &A\left(\frac{1}{\sqrt{M_k}}\sum_{i\in [N]}\left[\sqrt{c_{ik}}\ket{i,\xi_i^T}\otimes \left(\sqrt{\frac{c_{ik}}{c_i}}\ket{0}+\sqrt{1-\frac{c_{ik}}{c_i}}\ket{1}\right) \right]\right) \\
           =&\frac{1}{\sqrt{M_k}}\sum_{i\in [N]}\sqrt{c_{ik}}\ket{i,\xi_i^T,k}.
       \end{aligned}
         \label{eq:A act on state for k}
    \end{equation}
    Since $c_{ik}\leq c_i$, $A$ is well-defined. By direct calculation, it can be verified that the definition of $A$ preserves the inner product, so $A$ can be supposed as a unitary \cite{nielsen2010quantum}. By \Cref{eq:A*psi'}, we have
    \begin{align}
\sum_{T\in\mathcal{T}}\left|\braket{\psi_{t_k}|\tilde{\psi}^T}\right|\nonumber 
        =&\sum_{T\in\mathcal{T}}\left|\braket{\psi_{t_k},0|\tilde{\psi}^T,0}\right|\nonumber\\
        =&\sum_{T\in\mathcal{T}}\left|\bra{\psi_{t_k},0}A^\dagger \frac{1}{\sqrt{M}}\sum_{i\in[N]}\sum_{j\in[n]}\sqrt{c_{ij}}\ket{i,\xi_i^T,j}\right|\nonumber\\
        \leq& \begin{aligned}[t]
        \frac{1}{\sqrt{M}}&\sum_{T\in\mathcal{T}}\left|\sum_{i\in[N]}\sum_{j\neq k}\sqrt{c_{ij}}\bra{\psi_{t_k},0}A^\dagger\ket{i,\xi_i^T,j}\right|\\
        &+\frac{1}{\sqrt{M}}\sum_{T\in\mathcal{T}}\left|\sum_{i\in[N]}\sqrt{c_{ik}}\bra{\psi_{t_k},0}A^\dagger\ket{i,\xi_i^T,k}\right|.
        \end{aligned}\label{eqn:twoterm}
    \end{align}
    We then estimate the upper bound for these two terms in \Cref{eqn:twoterm} separately. 
    
    For the first term, note that $M_j=\sum_{i}c_{ij}$, by Cauchy-Schwartz inequality,
    \[
        \begin{aligned}
            &\frac{1}{\sqrt{M}}\sum_{T\in\mathcal{T}}\left|\sum_{i\in[N]}\sum_{j\neq k}\sqrt{c_{ij}}\bra{\psi_{t_k},0}A^\dagger\ket{i,\xi_i^T,j}\right|\\
            \leq& \frac{1}{\sqrt{M}}\sum_{T\in\mathcal{T}}\sqrt{\sum_{j\neq k}M_j}\sqrt{\sum_{i\in[N]}\sum_{j\neq k}\left|\bra{\psi_{t_k},0}A^\dagger\ket{i,\xi_i^T,j}\right|^2}
        \end{aligned}
    \]
    Since $A^\dagger$ is unitary, $\{A^\dagger\ket{i,\xi_i^T,j}\}_{i\in[N],j\in[n]}$ is an orthonormal system. Thus, the last square root is not greater than the norm of a unit vector $\ket{\psi_{t_k},0}$. So $\sqrt{\sum_{j\neq k}M_j/M}\cdot|\mathcal{T}|$ is an upper bound for the first term. 

    For the second term in \Cref{eqn:twoterm}, by \Cref{eq:A act on state for k} and the definition of the embedding map, we have 
    \begin{align*}
        &\frac{1}{\sqrt{M}}\sum_{T\in\mathcal{T}}\left|\sum_{i\in[N]}\sqrt{c_{ik}}\bra{\psi_{t_k},0}A^\dagger\ket{i,\xi_i^T,k}\right| \\
        =& \frac{1}{\sqrt{M}}\sum_{T\in\mathcal{T}}\left|\sum_{i\in[N]}\sqrt{\frac{c_{ik}}{c_i}}\sqrt{c_{ik}}\braket{\psi_{t_k}|i,\xi_i^T}\right| \\
        \leq& \frac{1}{\sqrt{M}}\sum_{T\in\mathcal{T}}\sum_{i\in[N]}\sqrt{c_{ik}}\left|\braket{\psi_{t_k}|i,\xi_i^T}\right|.
    \end{align*}
    Recall that $\kappa_k\geq c_{ik}$ and $c_{ik}=0$ for $i\notin T_k$, so the last expression is not greater than 
    \[\sqrt{\frac{\kappa_k}{M}}\sum_{T\in\mathcal{T}}\sum_{i\in T_k}\left|\braket{\psi_{t_k}|i,\xi_i^T}\right|.\] 
    Changing the summation order, we come to 
    \begin{align*}
    \sqrt{\frac{\kappa_k}{M}}\sum_{T\in\mathcal{T}}\sum_{i\in T_k}\left|\braket{\psi_{t_k}|i,\xi_i^T}\right|
    &=\sqrt{\frac{\kappa_k}{M}}\sum_{i\in[N]}\sum_{T\in\mathcal{T},i\in T_k}\left|\braket{\psi_{t_k}|i,\xi_i^T}\right|\\
    &\leq\sqrt{\frac{\kappa_k}{M}}\sum_{i\in[N]}\sum_{T\in\mathcal{T},i\in T_k}\left|\braket{\psi_{t_k}|i,\iota_i}\right| ,
    \end{align*}
    where "$\sum\limits_{T\in\mathcal{T},i\in T_k}$" means summation for all $T\in \mathcal{T}$ satisfying $i\in T_k$, and $\ket{\iota_i}$ independent to $T$ defined as \[|\braket{\psi_{t_k}|i,\iota_i}|=\max_{\|\ket{\xi}\|=1}|\braket{\psi_{t_k}|i,\xi}|.\] As $\ket{\psi_{t_k}}$ is also independent of the choice of the input $T\in\mathcal{T}$, the second summation just brings a multiplier, which is the count of $T\in \mathcal{T}$ satisfying $i\in T_k$. Since the definition of $\mathcal{T}$ is to choose $m_k$ elements in $[N]$ and assign them non-zero multiplicities, the count of the choices is $\binom{N-1}{m_k-1}$ as $i\in T_k$ means that $i$ has to been chosen. Then by Cauchy-Schwartz inequality and Lemma \ref{lem:size_of_T},
    \begin{align}
        &\sqrt{\frac{\kappa_k}{M}}\sum_{i\in[N]}\sum_{T\in\mathcal{T},i\in T_k}\left|\braket{\psi_{t_k}|i,\iota_i}\right| \notag \\
        =& \sqrt{\frac{\kappa_k}{M}}\binom{N-1}{m_k-1}\sum_{i\in[N]}\left|\braket{\psi_{t_k}|i,\iota_i}\right|\label{count}\\
        \leq& \sqrt{\frac{\kappa_k}{M}}\frac{m_k}{N}\binom{N}{m_k}\sqrt{N}\sqrt{\sum_{i\in[N]}\left|\braket{\psi_{t_k}|i,\iota_i}\right|^2}\notag\\
        \leq& \sqrt{\frac{\kappa_k}{MN}}m_k|\mathcal{T}|.\notag
    \end{align}

    Combining the above upper bounds for the two terms, we can obtain the proposition immediately.
\end{proof}

\begin{lemma}\label{F}\label{lem:lower bound for F_{t_k}}Let $\mathcal{T}$ be a collection generated by a hard input for $k\in[n]$ with constant $\beta$. If $M\leq\beta^2\kappa_k N/16$, then it holds that $F_{t_k}\geq M_k/2M$.
\end{lemma}
\begin{proof}
Noting that 
$\left\|\ket{\psi_{t_k}}-\ket{\tilde{\psi}^T}\right\|^2\geq 2-2\left|\braket{\psi_{t_k}|\tilde{\psi}^T}\right|$,
we have 
\[F_{t_k}\geq 2-\frac{2}{|\mathcal{T}|}\sum_{T\in\mathcal{T}}\left|\braket{\psi_{t_k}|\tilde{\psi}^T}\right|.\]
A bound for the summation is given by Proposition \ref{psi sum}. Through its conclusion, with $M=\sum_{j\in[n]}M_j$, we obtain a lower bound for $F_{t_k}$:
\begin{align*}
    F_{t_k}
    &\geq 2-2\sqrt{\frac{\sum_{j\neq k}M_j}{M}}-2\sqrt{\frac{\kappa_k}{MN}}m_k\\
    &= \frac{2}{1+\sqrt{\frac{\sum_{j\neq k}M_j}{M}}}\frac{M_k}{M}-2\sqrt{\frac{\kappa_k}{MN}}m_k\\
    &\geq \frac{M_k}{M}-2\sqrt{\frac{\kappa_k}{MN}}m_k.
\end{align*}
Recalling $\frac{M_k}{m_k}\geq \beta\kappa_k$ in the hard input conditions, and the condition of $M\leq \frac{\beta^2}{16}\kappa_k N$, we can bound the second term with \[\sqrt{\frac{\kappa_k}{MN}}\cdot m_k=\sqrt{\frac{M}{\kappa_kN}}\cdot \frac{\kappa_k m_k}{M}\leq \sqrt{\frac{\beta^2}{16}}\cdot\frac{M_k}{\beta M}= M_k/4M.\]
Thus, $F_t\geq M_k/2M$.
\end{proof}

To sum up, under the conditions of Lemma \ref{lem:lower bound for D_t}, by Lemma \ref{lem:upper bound for E_{t_k}} and Lemma \ref{lem:lower bound for F_{t_k}}, it follows that
$E_{t_k}\leq 2\epsilon, F_{t_k}\geq \frac{M_k}{2M}$.
Combining them, we immediately obtain  a lower bound for $\sqrt{F_{t_k}}-\sqrt{E_{t_k}}$:
\begin{align*}
    \sqrt{F_t}-\sqrt{E_t}
    &\geq\sqrt{\frac{M_k}{2M}}-\sqrt{2\epsilon}.
\end{align*}
The definition of hard input implies $M_k\geq\alpha M$. Since we have chosen $\alpha>4\epsilon$ as the condition of Lemma \ref{lem:lower bound for D_t}, it holds that $\epsilon\leq C_0M_k/M$ for some constant $C_0<1/4$. With the inequality \eqref{eq:lower bound for D_t by E_{t_k} and F_{t_k}}, the lower bound for $D_t$ can be obtained by
\[D_t\geq\left(\sqrt{F_t}-\sqrt{E_t}\right)^2\geq \left(\sqrt{\frac{M_k}{2M}}-\sqrt{2C_0\frac{M_k}{M}}\right)^2=C\frac{M_k}{M}\]
with constant $C=(1/\sqrt{2}-\sqrt{2C_0})^2$. Lemma \ref{lem:lower bound for D_t} follows from this.

\section{Proof of Lemma \ref{lem:upper bound for D_t}} \label{sec:proof of upper bound for D_t}

By the Cauchy-Schwartz inequality, 
\begin{align}
    D_{t+1}
    &=\frac{1}{|\mathcal{T}|}\sum_{T\in\mathcal{T}}\left\|U_{t+1}\left(O_{t+1}\ket{\psi_t^T}-\ket{\psi_t}\right)\right\|^2\notag\\
    &=\frac{1}{|\mathcal{T}|}\sum_{T\in\mathcal{T}}\left\|O_{t+1}\ket{\psi_t^T}-\ket{\psi_t}\right\|^2\notag\\
    &=\frac{1}{|\mathcal{T}|}\sum_{T\in\mathcal{T}}\left\|O_{t+1}(\ket{\psi_t^T}-\ket{\psi_t})+\left(O_{t+1}-I\right)\ket{\psi_t}\right\|^2\notag\\
    &\leq \begin{aligned}[t]
    &\frac{1}{|\mathcal{T}|}\sum_{T\in\mathcal{T}}\left\|\ket{\psi_t^T}-\ket{\psi_t}\right\|^2 \\
    &+\frac{2}{|\mathcal{T}|}\sum_{T\in\mathcal{T}}\left\|\ket{\psi_t^T}-\ket{\psi_t}\right\|\left\|\left(O_{t+1}-I\right)\ket{\psi_t}\right\|\\
    &+\frac{1}{|\mathcal{T}|}\sum_{T\in\mathcal{T}}\left\|\left(O_{t+1}-I\right)\ket{\psi_t}\right\|^2
    \end{aligned}\notag\\
    &\leq  \begin{aligned}[t]
    D_t+&2\sqrt{D_t}\left[\frac{1}{|\mathcal{T}|}\sum_{T\in\mathcal{T}}\left\|\left(O_{t+1}-I\right)\ket{\psi_t}\right\|^2\right]^{1/2} \\
    &+\frac{1}{|\mathcal{T}|}\sum_{T\in\mathcal{T}}\left\|\left(O_{t+1}-I\right)\ket{\psi_t}\right\|^2.
    \end{aligned}
    \label{D_t}
\end{align}

It remains to prove an upper bound for 
\[
\frac{1}{|\mathcal{T}|}\sum_{T\in\mathcal{T}}\left\|\left(O_{t+1}-I\right)\ket{\psi_t}\right\|^2.
\]
\begin{proposition}\label{sum O-I}
    For every collection $\mathcal{T}$ generated by a hard input for $k$, it holds that 
    \[\frac{1}{|\mathcal{T}|}\sum_{T\in\mathcal{T}}\left\|\left(O_{t+1}-I\right)\ket{\psi_t}\right\|^2\leq 4\frac{m_k}{N}.\]
\end{proposition}
\begin{proof}
    Since $O_{t+1}$ is either $\hat{\mathcal{O}}_k\otimes I$ or $\hat{\mathcal{O}}_k^\dagger\otimes I$, by the definition of $\hat{\mathcal{O}}_k$, we have
    \begin{align*}
        &\left(O_{t+1}-I\right)\ket{\psi_t} \\
        =&\sum_{i\in T_k}\sum_{s=0}^\nu \sum_l\ket{i,s,1,l}\left(\bra{i,s\oplus \pm c_{ik}, 1,l}-\bra{i,s,1,l}\right)\ket{\psi_t} 
    \end{align*}
    with $x\oplus y:=(x+y)\mod (\nu+1)$.
    
    Notice that for two complex numbers $a,b$, it holds that $|a-b|^2\leq 2(|a|^2+|b|^2)$. Hence,
    \begin{align*}
        &\sum_{T\in\mathcal{T}}\left\|\left(O_{t+1}-I\right)\ket{\psi_t}\right\|^2 \\
        =& \sum_{T\in\mathcal{T}}\left\|\sum_{i\in T_k}\sum_{s=0}^\nu \sum_l\ket{i,s,1,l}\left(\bra{i,s\oplus \pm c_{ik}, 1,l}-\bra{i,s,1,l}\right)\ket{\psi_t} \right\|^2 \\
        =&\sum_{T\in\mathcal{T}}\sum_{i\in T_k}\sum_{s=0}^\nu \sum_l \left|\left(\bra{i,s\oplus \pm c_{ik}, 1,l}-\bra{i,s,1,l}\right)\ket{\psi_t}\right|^2 \\
        \leq& 2\sum_{T\in\mathcal{T}}\sum_{i\in T_k}\sum_{s=0}^\nu\sum_l\left(|\braket{i,s\oplus \pm c_{ik},1,l|\psi_t}|^2+|\braket{i,s,1,l|\psi_t}|^2\right)\\
        =& 4\sum_{T\in\mathcal{T}}\sum_{i\in T_k}\sum_{s=0}^\nu\sum_l|\braket{i,s,1,l|\psi_t}|^2.
    \end{align*}
    Similarly to the deduction of \Cref{count}, through changing the summation order, since $\ket{\psi_t}$ is invariant across every $T\in\mathcal{T}$, and the count of $T\in\mathcal{T}$ satisfying $i\in T_k$ is $\binom{N-1}{m_k-1}$, it can be concluded that
    \begin{align*}
        &4\sum_{T\in\mathcal{T}}\sum_{i\in T_k}\sum_{s=0}^\nu\sum_l|\braket{i,s,1,l|\psi_t}|^2\\
        =& 4\sum_{s=0}^\nu\sum_l\sum_{i\in [N]}\sum_{T\in\mathcal{T},i\in T_k}|\braket{i,s,1,l|\psi_t}|^2\\
        =& 4\binom{N-1}{m_k-1}\sum_{s=0}^\nu\sum_l\sum_{i\in [N]}|\braket{i,s,1,l|\psi_t}|^2\\
        \leq& 4\binom{N-1}{m_k-1}=4\frac{m_k}{N}\binom{N}{m_k}.
    \end{align*}
    Thus, the proposition can be obtained immediately from this with \Cref{lem:size_of_T}.
\end{proof}

By Proposition \ref{sum O-I} and the inequality \eqref{D_t}, we have the relationship
\[D_{t+1}\leq D_t+4\sqrt{\frac{m_k}{N}D_t}+4\frac{m_k}{N}=\left(\sqrt{D_t}+2\sqrt{\frac{m_k}{N}}\right)^2.\]
This implies the upper bounds for $\sqrt{D_t}$ form an arithmetic progression.
Since $D_0=0$, we can obtain that 
$ D_t\leq \left(2\sqrt{\frac{m_k}{N}}t\right)^2= 4\frac{m_k}{N}t^2$.

\end{document}